\documentclass{article}
\usepackage[english]{babel}

\usepackage[margin=1in]{geometry}

\usepackage{tabularx} 

\newcolumntype{Y}{>{\centering\arraybackslash}X}

\usepackage[utf8]{inputenc} 
\usepackage[T1]{fontenc}    
\usepackage{hyperref}       
\usepackage{url}            
\usepackage{booktabs}       
\usepackage{amsfonts}       
\usepackage{nicefrac}       
\usepackage{microtype}      
\usepackage{lipsum}		
\usepackage{graphicx}
\usepackage[numbers]{natbib}
\usepackage{doi}
\usepackage{amsmath}
\usepackage{mathtools}
\usepackage{tabularx,ragged2e,booktabs,caption}
\usepackage{amsthm}
\newtheorem{theorem}{Theorem}[section]

\newtheorem{lemma}[theorem]{Lemma}
\newtheorem{definition}{Definition}[section]
\usepackage{amssymb}
\usepackage{amsmath}
\usepackage{commath}
\usepackage{tikz}
\usetikzlibrary{positioning}

\newcommand{\mycomment}[1]{}
\usepackage{hyperref}
\usepackage[skip=0.5\baselineskip]{parskip}
\usepackage{multirow}
\usepackage{multirow}
\usepackage{tablefootnote}

\usepackage{cleveref}
\crefname{remark}{rmk}{rmks}
\usepackage{subcaption}

\usepackage{indentfirst}
\usepackage[margin=1in]{geometry}

\usepackage[utf8]{inputenc} 
\usepackage[T1]{fontenc}    
\usepackage{hyperref}       
\usepackage{url}            
\usepackage{booktabs}       
\usepackage{amsfonts}       
\usepackage{nicefrac}       
\usepackage{microtype}      
\usepackage{lipsum}		
\usepackage{graphicx}
\usepackage[numbers]{natbib}
\usepackage{doi}
\usepackage{mathtools}
\usepackage{amsmath}
\usepackage{tabularx,ragged2e,booktabs,caption}
\usepackage{tikz}
\usetikzlibrary{positioning}
\usepackage[makeroom]{cancel}
\usetikzlibrary{shapes,arrows}
\usepackage{caption}
\usepackage{booktabs} 
\usepackage{siunitx}  

\usetikzlibrary{matrix,shapes,arrows,positioning,chains}

\usepackage{bigints}

\DeclarePairedDelimiter{\ceil}{\lceil}{\rceil}

\tikzset{
decision/.style={
    ellipse,
    draw,
    text width=10em,
    text badly centered,
    inner sep=3pt
},
block/.style={
    rectangle,
    draw,
    text width=14em,
    text centered,
    rounded corners
},
cloud/.style={
    draw,
    ellipse,
    minimum height=2em
},
descr/.style={
    fill=white,
    inner sep=2.5pt
},
connector/.style={
    -latex,
    font=\scriptsize
},
rectangle connector/.style={
    connector,
    to path={(\tikztostart) -- ++(#1,0pt) \tikztonodes |- (\tikztotarget) },
    pos=0.5
},
rectangle connector/.default=-2cm,
straight connector/.style={
    connector,
    to path=--(\tikztotarget) \tikztonodes
}
}

\makeatletter
\newcommand*{\centernot}{%
  \mathpalette\@centernot
}
\def\@centernot#1#2{%
  \mathrel{%
    \rlap{%
      \settowidth\dimen@{$\m@th#1{#2}$}%
      \kern.5\dimen@
      \settowidth\dimen@{$\m@th#1=$}%
      \kern-.5\dimen@
      $\m@th#1\not$%
    }%
    {#2}%
  }%
}
\makeatother

\newcommand{\independent}{\perp\mkern-9.5mu\perp}

\date{} 					
\begin{document}
\title{Non-parametric Replication of Instrumental Variable Estimates Across Studies}

\author{{Roy S. Zawadzki}\thanks{Corresponding Author: zawadzkr@uci.edu}\\
	Department of Statistics\\
	University of California, Irvine\\
	Irvine, CA\\
	\and
	{Daniel L. Gillen}\\
	Department of Statistics\\
	University of California, Irvine\\
	Irvine, CA\\
        \and
	for the Alzheimer’s Disease Neuroimaging Initiative\thanks{Data used in preparation of this article were obtained from the Alzheimer’s Disease
Neuroimaging Initiative (ADNI) database (adni.loni.usc.edu). As such, the investigators
within the ADNI contributed to the design and implementation of ADNI and/or provided data
but did not participate in analysis or writing of this report. A complete listing of ADNI
investigators can be found at:
\href{http://adni.loni.usc.edu/wp-content/uploads/how_to_apply/ADNI_Acknowledgement_List.pdf}{http://adni.loni.usc.edu/wp-content/uploads/how\_to\_apply/ADNI\_Acknowledgement\_List.pdf}}
}


\maketitle

\begin{abstract}
\noindent Replicating causal estimates across different cohorts is crucial for increasing the integrity of epidemiological studies. However, strong assumptions regarding unmeasured confounding and effect modification often hinder this goal. By employing an instrumental variable (IV) approach and targeting the local average treatment effect (LATE), these assumptions can be relaxed to some degree; however, little work has addressed the replicability of IV estimates. In this paper, we propose a novel survey weighted LATE (SWLATE) estimator that incorporates unknown sampling weights and leverages machine learning for flexible modeling of nuisance functions, including the weights. Our approach, based on influence function theory and cross-fitting, provides a doubly-robust and efficient framework for valid inference, aligned with the growing "double machine learning" literature. We further extend our method to provide bounds on a target population ATE. The effectiveness of our approach, particularly in non-linear settings, is demonstrated through simulations and applied to a Mendelian randomization analysis of the relationship between triglycerides and cognitive decline. 

\end{abstract}

\section{Introduction}

In recent years, researchers have come to recognize that science faces a replicability crisis.\citep{pashler_is_2012,ioannidis_why_2005} That is, many findings produced by one study are unable to be repeated in subsequent similar studies. The end result is a potential reduction in the scientific community's confidence in published results. In the context of medical research, many have highlighted instances of limited replicability. \citep{ioannidis_contradicted_2005,begley_raise_2012,mathur_toward_2023} Epidemiologists frequently undertake questions for which controlled experimental data cannot be feasibly produced. Among other factors, this leads to a chief concern of spurious results due to unmeasured confounding, further highlighting the importance of replicability across multiple cohorts to increase confidence in conclusions. As an illustrative example, consider the effect of elevated blood lipid levels on cognitive decline. Though some work has found a that elevated lipid levels are associated with Alzheimer's Disease (AD) incidence, others have reported null results.\citep{solomon_midlife_2009,kivipelto_apolipoprotein_2002,tan_plasma_2003,mielke_32-year_2010} 

One may hypothesize two natural reasons that explain the conflicting results in our example: unmeasured confounding and unaccounted for effect modification across studied cohorts. While much of the work on observational study methodology has centered around the former issue, there is comparatively less work that focuses on the latter. To relax the assumption that all confounders, or suitable proxies, were measured, we may turn to an instrumental variable (IV) approach via Mendelian randomization (MR). Briefly, MR involves using one or more single nucleotide polymorphism (SNP) that meets three main conditions: (i) it influences the level of triglycerides (relevance), (ii) it does not directly cause the cognitive decline (exclusion restriction), and (iii) there is no association any unobserved confounders (independence).\citep{baiocchi_tutorial_2014}. For example, an analysis by Proitsi et al. (2014) found a null effect between triglycerides and late onset AD via MR.\citep{proitsi_genetic_2014} Yet, considering the previous literature on this topic, it is natural to ask whether such findings can be replicated in other settings. This crucially involves a nuanced discussion surrounding how the conditions for replicability interacts with the causal assumptions required for internal validity.

Suppose we wish to compare causal estimates across two independent studies. This requires us to properly account for differences in the distribution of effect modifiers between each study such as medical history, concomitant medications, and lifestyle factors with many of these simultaneously being confounders. Traditionally, this may be done via survey weighting methodology. \citep{cole_generalizing_2010,ridgeway_propensity_2015,bernstein_adjustment_2023}. Nevertheless, there are three assumptions for such an approach to succeed: (i) all possible effect modifiers are measured, (ii) the distributions of these effect modifiers overlap (i.e. positivity), and (iii) the sampling weights are properly modeled. In many ways, these assumptions are similar to that of propensity score (PS) analyses in which we can require only balancing a correctly specified PS comprised of variables we are able to measure and that, furthermore, sufficiently overlap between treatment groups.\citep{rosenbaum_central_1983,brooks_squeezing_2013} Similar to the threats to internal validity posed by untestable causal assumptions, the threats to replicability are conceptually alike.

One path forward to relax the conditions required for sample weighting methods is by narrowing the scope of the target estimand from the whole population (i.e. the ATE) to a subset. By restricting the population of interest to be more homogeneous, the variability of all possible effect modifiers decreases, including those unobserved. Furthermore, the area of overlap is more likely to be proportionally larger. Through estimation of the local average treatment effect (LATE),\citep{imbens_identification_1994} an IV approach with a valid instrument will achieve the same effect while mitigating unobserved confounding. The aforementioned "narrowing" occurs because the IV approach identifies the causal effect of a “complier” population, defined as those whose level of exposure monotonically varies with the value of the IV. This means that unless we assume treatment effect homogeneity or that treatment effect heterogeneity is unrelated to treatment assignment, the LATE does not equal the ATE.\citep{wang_bounded_2018,hartwig_average_2023} Even so, the IV approach could allow us to replicate causal estimates with increased confidence. Unlike the ATE, however, there has been little research regarding how the LATE may be used for replicability.

In the current work, assuming a binary IV and treatment, we extend LATE estimators to incorporate unknown sampling weights into a weighted LATE (WLATE). The WLATE seeks to generalize the results from one study to the target distribution of another, thus allowing for the basis of replication of effect estimates across studies. While existing WLATE methods could be employed to provide analysts valid point estimates, they require one to assume the sampling weights are known for valid inference and are further restrictive to which models can be used to estimate nuisance functions such as the weights.\citep{choi_instrumental_2023} Using the theory of influence functions (IFs) and sample-splitting,\citep{kennedy_semiparametric_2023} our proposed method both accounts for the uncertainty in estimating the sample weights using a straightforward variance estimator and allows analysts to utilize machine learning (ML) to flexibly model nuisance functions in order to guard against potential misspecification. Furthermore, our estimator is doubly-robust with respect to within-study nuisance functions. 

The remainder of this work is organized as follows. First, we describe the necessary conditions and key estimand of our methodology that we call the survey weighted LATE (SWLATE). Then, we develop an IF-based approach and describe a novel across-study cross-fitting procedure to estimate the SWLATE. To our knowledge, there is no result deriving the IF and showing the regularity conditions and asymptotic properties of a weighted LATE when using ML to estimate unknown sampling weights. Following this, to ameliorate concerns over the narrow scope of the LATE, we extend our developed method to estimate bounds on the weighted ATE (WATE) based on the approach of Kennedy et al. (2020).\citep{kennedy_sharp_2020} Lastly, we demonstrate our estimators for the SWLATE and WATE both in simulation and in an MR analysis of the effect of triglycerides on cognitive decline.

\section{Notation and Set-Up}

Notationally, we assume a binary treatment $D \in \{0,1\}$ and a binary IV $Z \in \{0,1\}$. $Y(d)$ is the potential outcome under treatment assignment $d$ and $D(z)$ as the potential treatment assignment under IV assignment of $z$ and therefore $Y = DY(1) + (1-D)Y(0)$ and $D = ZD(1) + (1-Z)D(0)$. Denote $Y(d,z)$ as the potential outcome under treatment assignment $d$ and IV assignment $z$. We will assume $Y$ is continuous throughout this work. We additionally assume the standard assumption of the stable unit treatment value assumption (SUTVA) for all potential outcomes. Formally, the LATE and ATE are define as $\beta_{ATE} = E[Y(1) - Y(0)]$ and $\beta_{LATE} = E[Y(1) - Y(0)|D(1) > D(0)]$ where $D(1) > D(0)$ refers to the "principal strata" of compliers or those whose treatment status vary with the IV.\citep{angrist_identification_1996} The other principal strata are always-takers, $D(1) = D(0) = 1$, never takers, $D(1) = D(0) = 0$, and defiers, $D(1) < D(0)$.

A set of covariates $X$ may be required to satisfy the IV independence assumption.\citep{frolich_nonparametric_2007,abadie_semiparametric_2003} As such, we utilize the instrument propensity score (IPS) $e(X) = P(Z = 1|X)$. The conditions to identify the LATE via the IPS are similar to the that of ATE: positivity of the IPS and for the set of confounders to achieve strong ignorability. Additionally, conditioning on $X$ means the standard IV assumptions of relevance and monotonicity are strata-specific. Formally, we may define the assumptions to identify the LATE conditional as follows:
\begin{align*}
    &\textbf{Assumption 1 (Positivity of IPS): } 0 < e(X) < 1 \text{ a.s. for all } x \in \mathcal{S}_X \nonumber \\
    &\textbf{Assumption 2 (Independence): } Z \independent \{Y(0,0), Y(0,1), Y(1,0), Y(1,1), D(1), D(0)\} \,|\, X \nonumber \\
    &\textbf{Assumption 3 (Exclusion Restriction): } Y(d,0) = Y(d,1) \text{ for } d \in \{0,1\} \nonumber \\
    &\textbf{Assumption 4 (Relevance): } P[D(1) = 1|X] > P[D(0) = 1|X] \text{ a.s. for all } x \in \mathcal{S}_X \nonumber \\
    &\textbf{Assumption 5 (Monotonicity): } P[D(1) \ge D(0)|X] = 1 \text{ a.s. for all } x \in \mathcal{S}_X \nonumber
\end{align*}
\noindent where $\mathcal{S}_X$ refers to the support of $X$. In order to formulate the SWLATE, we include weights $w(X)$ in the LATE estimand outlined in Frolich (2007) Theorem 1, which leads us to the following functional form
\begin{equation}\label{eq:wlate}
    \beta_{SWLATE} = \frac{\int\big[\mu_1(x) - \mu_0(x)\big]w(x)f_X(x)dx}{\int\big[m_1(x) - m_0(x)\big]w(x)f_X(x)dx}
\end{equation}
\noindent where $\mu_z(X) = E[Y|X,Z = z]$ and $m_z(X) = E[D|X,Z = z]$. To define $w(X)$ for the purposes of replicability, we must outline some further notation and conditions. 

Let $A = (X^A,D^A,Y^A,Z^A)$ represent a sample of size $N_A$ from study A with distribution $P_A$ and $B = (X^B,D^B,Y^B,Z^B)$ a sample of size $N_B$ from study B with distribution $P_B$. We assume that all variables were collected in the same manner and that Assumptions 1-5 are met across both studies. Assuming that $P_A \ne P_B$, for example due to differences in effect modifiers between the samples, the study-specific LATEs, $\beta^A_{LATE}$ and $\beta^B_{LATE}$, are not equal. For the purposes of replicability, this is the crux of the problem and may be ameliorated via modifying one of the LATEs (we arbitrarily choose study B) via sampling weights. When this is the case, our target estimand is $\beta^A_{LATE}$ and we denote the SWLATE as $\beta^B_{SWLATE}$ because it is estimated over $P_B$. Thus Eq. \ref{eq:wlate} can be rewritten as
\begin{align}\label{eq:wlate_2}
    \beta^B_{SWLATE} = \frac{E_{P_B}[w(X)\{\mu_1(X) - \mu_0(X)\}]}{E_{P_B}[w(X)\{m_1(X) - m_0(X)\}]}.
\end{align}

Let $M_A$ and $M_B$ represent indicators of membership for study A and B, respectively. We construct weights based upon the probability of being sampled into study B: $\eta(X) = P(M_B = 1|X)$. Typically, these probabilities are unknown and must be estimated by constructing and fitting a model for $E[M_B|X]$ using the data from both $A$ and $B$. Consequentially, we must account for the uncertainty in estimating $\eta(X)$ in inference on $\beta^B_{SWLATE}$.

We assume the following conditions for our sampling weights: the subjects have a non-zero probability of being sampled into each study, both studies' conditional LATEs are equal within each possible strata of the covariates (i.e. no hidden strata-specific treatment effect heterogeneity), and that $\eta(X)$ is properly specified. Formally, letting $\mathcal{S}_X = \mathcal{S}^A_X \cup \mathcal{S}^B_X$ we have:
\begin{align*}
    &\textbf{Assumption 6 (Positivity of Sampling Weights): } 0 < \eta(X) < 1 \text{ a.s. for all } x \in \mathcal{S}_X \nonumber \\
    &\textbf{Assumption 7 (Equality of Strata Specific LATEs): }  \beta(x)^A_{LATE} = \beta(x)^B_{LATE} \text{ a.s. for all } x \in \mathcal{S}_X \nonumber\\
    &\textbf{Assumption 8 (Proper Specification of the Weights): } |\hat{\eta}(X)-\eta(X)| = o_p(1)
\end{align*}
\noindent With these conditions, we may utilize $w(X) = \frac{1-\eta(X)}{\eta(X)}$ to re-weight $\beta^B_{LATE}$ to replicate $\beta^A_{LATE}$. Assumption 7 encapsulates the previous argument that we may relax assumptions regarding unmeasured effect modification and positivity for the overall population by replacing them with similar assumptions for the complier population whose distribution of covariates varies less across studies.

\section{Nonparametric Estimation of the SWLATE with Unknown Sampling Weights}

To estimate Eq. \ref{eq:wlate_2}, we extend the IF and cross-fitting framework for the LATE as detailed in Kennedy (2023)\citep{kennedy_semiparametric_2023} to add survey weights $w(X)$. This allows us to non-parametrically identify and estimate the SWLATE with the use of ML (e.g. regularization, random forests, neural networks). Additionally, our estimator admits double-robustness for estimating the LATE within each study if either $m_1(X), m_0(X)$, $\mu_1(X)$, and $\mu_0(X)$ or $e(X)$ are correctly specified, which assists replicability by relaxing the assumptions for internal validity. When survey weights are unknown, our application of IF theory and sample-splitting both avoids complex derivations of asymptotic results using many simultaneous estimating equations and allows for flexible modeling of $w(X)$, protecting against potential misspecification.

In the remainder of the text, the empirical measure $\mathbb{P}_n$ refers to a sub-sample of $B$ created by sample-splitting. For example, suppose the data has been split into two partitions, $B^n$ and $B^N$, with $n = \ceil{\frac{N_B}{2}}$. Letting $\zeta_z = (\mu_z\{X\},m_z\{X\},e\{X\})$ represent the nuisance functions associated with study B, we can estimate $\zeta_z$ using $B^N$. Following this, we may fit $\eta(X)$ with both $B^N$ and $A$. Computing the predicted values for each nuisance function using $B^n$, we may adapt the estimand presented in Mao et al. (2019) Appendix Theorem 2\citep{mao_propensity_2019} to the following plug-in estimators for the uncentered IFs of the numerator and denominator of Eq. \ref{eq:wlate}:
\begin{align}
    \phi_{num}(B;\hat{w},\hat{\zeta_z}) &= \frac{\hat{w}(X)}{\mathbb{P}_n\{\hat{w}(X)\}}\bigg[\frac{Z}{\hat{e}(X)}\{Y-\hat{\mu}_1(X)\} - \frac{1-Z}{1-\hat{e}(X)}\{Y-\hat{\mu}_0(X)\} + \hat{\mu}_1(X) - \hat{\mu}_0(X)\bigg] \label{eq:num_estm} \\
    \phi_{denom}(B;\hat{w},\hat{\zeta_z}) &= \frac{\hat{w}(X)}{\mathbb{P}_n\{\hat{w}(X)\}}\bigg[\frac{Z}{\hat{e}(X)}\{D-\hat{m}_1(X)\} - \frac{1-Z}{1-\hat{e}(X)}\{D-\hat{m}_0(X)\} + \hat{m}_1(X) - \hat{m}_0(X)\bigg] \label{eq:denom_estm}
\end{align}

\noindent with point estimate $\hat{\beta}^B_{SWLATE} = \mathbb{P}_n\{\phi_{num}(B;\hat{w},\hat{\zeta_z})\}\big/\mathbb{P}_n\{\phi_{denom}(B;\hat{w},\hat{\zeta}_z)\}$. We define the following cross-fitting procedure involving data from two studies and use this to formulate the main result of the paper, Theorem \label{theorem_asymptotic}.

\begin{definition}[Cross-fitting Procedure for Two Studies]\label{prop:cross-fit}
To estimate $\hat{\beta}^B_{SWLATE}$ via cross-fitting we execute the following procedure:
\begin{enumerate}
    \item Randomly split $B$ into $K$ folds equally and $N_{B,k}$ sized. Let $k \in \{1,...K\}$ be an arbitrary fold. Define $I_{B,k}$ as the indices in $k$ and, similarly, $I^{c}_{B,k} := \{1...N_B\} \setminus I_{B,k}$.
    \item For each $k$:
    \begin{enumerate}
        \item Combine datasets into $C = (A, (B_i)_{i \in I^{c}_{B,k}})$.
        \item Use $C$ to estimate $\eta_{-k}$ where $-k$ refers to omitting the fold $k$.
        \item Use $(B_i)_{i \in I^{c}_{B,k}}$ to estimate the study-specific nuisance functions $\zeta_{z,-k}$
        \item With the "hold-out" sample $(B_i)_{i \in I_{B,k}}$, compute the "predictions" $\hat{w}_k(X)$ and $\hat{\zeta}_{z,k}$.
        \item Plug in these values into the estimators in Equations \ref{eq:num_estm} and \ref{eq:denom_estm} to obtain ${\hat{\beta}^B_{k,SWLATE}}$.
    \end{enumerate}
    \item Average across the folds to obtain the final estimate: $\hat{\beta}^B_{SWLATE} = \frac{1}{K}\sum_{k=1}^{K}{\hat{\beta}^B_{k,SWLATE}}$.
\end{enumerate}
\end{definition}

\begin{theorem}[Inference for the SWLATE via the Influence Function]\label{theorem_asymptotic}
 The following conditions must be met:
\begin{enumerate}
    \item $N_A/N_B \overset{p}{\to} c > 0$ where $c$ is some constant.
    \item Assumption 10 holds as well as $|\hat{\mu}_z - \mu_z| = o_P(1)$, $|\hat{m}_z - m_z| = o_P(1)$, and $|\hat{e}(X) - e(X)| = o_P(1)$ for $z \in \{0,1\}$.
    \item $\norm{\hat{e} - e}\underset{z \in \{0,1\}}{\max}(\norm{\hat{\mu_z} - \mu_z},\norm{\hat{m}_z - m_z}) = o_p(N_B^{-1/2})$
    \item $|Y| < C_Y$ a.s. for some constant $C_Y$
    \item $|\mu_z| < C_{\mu_z}$ and $m_z < C_{m_z}$ a.s. where $C_{\mu_z}$ and $C_{m_z}$ are some constants for $z \in \{0,1\}$
    \item $\hat{e},e > C_{e}$ and $\hat{\eta},\eta > C_{\eta}$ a.s. where $C_{e}$ and $C_{\eta}$ are some constants such that $C_{e}, C_{\eta} > 0$ 
    \item $\frac{\hat{w}(X)}{\mathbb{P}_n\{\hat{w}(X)\}} < C_{w}$ a.s where $C_w$ is some constant
\end{enumerate}
\noindent Using the estimator described in Definition \ref{prop:cross-fit}, we have that $N_B^{1/2}(\hat{\beta}^B_{SWLATE} - \beta^A_{LATE}) \overset{d}{\to} N(0,E[\Gamma^2])$ where
\begin{align*}
    \Gamma &= \frac{w(X)}{E_{P_B}[w(X)\{m_1(X) - m_0(X)\}]}\bigg\{\frac{2Z-1}{e(X,Z)}\bigg[Y-\mu_Z(X)-\beta^A_{LATE}\{D-m_Z(X)\}\bigg]\\
    &+ \mu_1(X) - \mu_0(X) - \beta^A_{LATE}\{m_1(X) - m_0(X)\}\bigg\}
\end{align*}
\noindent with $e(X,Z) = e(X)Z + \{1-e(X)\}(1-Z)$.
\end{theorem}

From the results of Theorem \ref{theorem_asymptotic}, we may construct $(1-\alpha)$-level Wald confidence intervals for $\hat{\beta}^B_{SWLATE}$. Letting $q_{1-\alpha/2}$ be the $1-\alpha/2$ quantile of the standard normal distribution and $\hat{\Gamma}$ being the plug-in estimator for $\Gamma$, we have:
\begin{equation}
    \hat{\beta}^B_{SWLATE} \pm q_{1-\alpha/2}\sqrt{\frac{\mathbb{P}_{N_B}\hat{\Gamma}^2}{N_B}}.
\end{equation}

A notable property of our estimator is that by taking $E_{P_B}[\Gamma^2]$, by iterated expectation on $X$ and $Z$, we obtain the efficiency bound for the WLATE as given by Choi (2023) Theorem 1.\citep{choi_instrumental_2023} That is, our estimator achieves the lowest possible variance across all possible non-parametric WLATE estimators.\citep{kennedy_semiparametric_2023} In contrast with Choi (2023), we need not know the weights nor the IPS to achieve this lower bound. Altogether, this means that we may efficiently, non-parametrically, and flexibly compute a weighted LATE estimator to replicate findings from IV analyses across studies with all nuisance functions unknown.

\noindent 

\section{Weighted ATE Bounds with Unknown Sampling Weights}

Without restrictive and untestable assumptions regarding treatment effect heterogeneity, the LATE will not identify the ATE.\citep{hartwig_average_2022,wang_bounded_2018} Nevertheless, we may use the LATE to bound the ATE, which was first developed for non-compliance in clinical trials.\citep{manski_nonparametric_1990,robins_correcting_1994,balke_bounds_1997} Kennedy et al. (2020) generalized this to tighter bounds that allows conditioning on covariates tha we will utilize.\citep{kennedy_sharp_2020}

Let $Y$ be bounded and scaled such that $Y \in [0,1]$, then for $j \in \{l,u\}$, referring to the lower and upper bounds, Kennedy et al. (2020) defines the bound as
\begin{align}\label{eq:upper_bound_gen}
    \beta^j_{ATE} = E[E[V_{j,1}|X,Z = 1] - E[V_{j,0}|X, Z = 0]] 
\end{align}
for $V_{u,1} = YD + 1 - D$, $V_{l,1} = YD$, $V_{u,0} = Y(1-D)$, and $V_{l,0} = Y(1-D) + D$. The width of the interval $(\beta^l_{ATE},\beta^u_{ATE})$ is inversely proportional to the strength of the instrument: stronger instruments will yield smaller intervals. Using the results from the previous section, we can straightforwardly extend this approach to incorporate survey weights.


Letting $\zeta_{j,z} = (v_{j,z},e\{X\})$, the uncentered IF $\phi_{ATE,j}(B;w,\zeta_{j,z})$ takes the following form
\begin{equation}
    \phi_{ATE,j}(B;\hat{w},\hat{\zeta}_{j,z}) =  \frac{\hat{w}(X)}{\mathbb{P}_n\{\hat{w}(X)\}}\bigg[\frac{Z}{\hat{e}(X)}\{V_{j,1}-\hat{v}_{j,1}(X)\} - \frac{1-Z}{1-\hat{e}(X)}\{V_{j,0}-\hat{v}_{j,0}(X)\}
    + \hat{v}_{j,1}(X)- \hat{v}_{j,0}(X)\bigg]. \label{eq:weighted_estm_j} 
\end{equation}

\noindent Denoting our study weighted estimate for the bounds ATE of study A as $\hat{\beta}^B_{SWATE,j}$, we denote $\hat{\beta}^B_{SWATE,l} = \mathbb{P}_n\{\phi_{ATE,l}(B;\hat{w},\hat{\zeta}_{l,z})\}$ and $\hat{\beta}^B_{SWATE,u} = \mathbb{P}_n\{\phi_{ATE,u}(B;\hat{w},\hat{\zeta}_{u,z})\}$ as the lower and upper bounds, respectively. Each weighted bound can be fit with the procedure described in Definition \ref{prop:cross-fit}. Furthermore, inference follows from conditions identical to Theorem $\ref{theorem_asymptotic}$ since each bound is a essentially weighted ATE, we may use the results in the proof for the SWLATE numerator. Standard $1-\alpha$ level confidence intervals can be computed as
\begin{align}
    \hat{\beta}^B_{SWATE,j} \pm q_{1-\alpha/2}\sqrt{\frac{\mathbb{P}_{N_B}\hat{\Delta}_j^2}{N_B}}
\end{align}
\noindent where $\hat{\Delta}_j$ is the plug in estimator for the centered influence function, which takes the form
\begin{align}
    \Delta_j = \frac{w(X)}{E_{P_B}[w(X)]}\bigg\{\frac{2Z-1}{\hat{e}(X,Z)}[V_{j,z}-v_{j,z}(X)] + v_{j,1}(X) - v_{j,0}(X)\bigg\} - \beta^A_{ATE,j}.
\end{align}
\noindent From these results, analysts can compute bounds on the ATE based on the LATE that are weighted towards a target study population and, if desired, compute confidence intervals on these weighted bounds.

\section{Simulation}

We consider two general scenarios: a linear and non-linear data-generating mechanism (DGM). First, we generate six continuous covariates $X = (X_1,X_2,...,X_6)$, from a multivariate normal distribution with $N = 1500$. Each variable has mean 0 and variance 1.5 with two correlated blocks, $X_1,X_2,X_3$ and $X_4,X_5,X_6$, with a covariance of 0.3. Table \ref{tab:dgm_sim} defines all equations and coefficients for both the the linear and non-linear DGMs. First, we generate sampling probabilities of being in study $A$ with $P(M_B = 1) = 0.5$ and use these probabilities to sample 1500 observations with replacement form study $B$'s sample. Within each dataset, we assign the IV using the IPS with $P(Z=1)=0.5$. 

Following this, we generate compliance status based on the "compliance score," or probability of being a complier. We set $P(D(1) > D(0))$ at 0.2, 0.5, and 0.8 to represent a "Weak IV","Moderate IV", and "Strong IV", respectively. Letting $\delta(X) = P(D(1) > D(0)|X)$, we calculate $P(D(1) = D(0) = 1|X) = P(D(1) = D(0) = 0|X) = \frac{1-\delta(X)}{2}$ for each subject, or the probabilities of being an always-taker and never-taker, respectively. From these probabilities, a subject's principal strata is selected via a multinomial distribution where $S=0$ if a never-taker, $S=1$ if an always-taker, and $S=2$ if a complier. For subject $i$, the treatment assignment is $D_i = Z_iI(S_i=2) + S_iI(S_i \ne 2)$. 

In the outcome DGM, we set the Study A ATE (i.e. $\beta_1$) at 1. Our formulation in Table \ref{tab:dgm_sim} allows us to write the LATE in closed form as $\beta_{LATE}^A = \beta_1 + \sum_{i=2}^{7}\beta_iE[X_i|S = 2]$ where we obtain the ground-truth LATEs via averaging across 5000 simulations to estimate $E[X_i|S = 2]$ where we know the complier status. Notationally, we let $\beta_{LATE,L}^A$ refer to the LATE in study A under the linear DGM while $\beta_{LATE,NL}^A$ refers to the non-linear DGM. The same interpretation applies to $\beta_{LATE,L}^B$ and $\beta_{LATE,NL}^B$. 

To estimate the nuisance functions, we will compare utilizing generalized linear models (GLMs) to an ensemble of flexible models via the \texttt{SuperLearner} package.\citep{polley_superlearner_2024} For the ensemble, we will we use multivariate adaptive regression splines (EARTH), generalized additive models (GAMs), GLMs, random forest (RF), and recursive partitioning and regression trees (RPART). "SL" or "No SL" denotes whether the ensemble SuperLearner was used or not. Nuisance functions were estimated with cross-fitting with four splits of the data.

\begin{table}[]
    \centering
    \caption{Summary of Data-Generating Mechanisms for Simulation}
    \label{tab:dgm_sim}
    \begin{tabularx}{\textwidth}{|l|l|>{\raggedright\arraybackslash}p{6.3cm}|>{\raggedright\arraybackslash}p{4cm}|}
        \toprule
        \textbf{Model} & \textbf{DGM} & \textbf{Equation} & \textbf{Coefficient Values} \\
        \midrule
        Sampling Probability & Linear & $logit[P(M_B = 1|X=x)] = \alpha_0 + \sum_{i=1}^{6}\alpha_iX_i$ & $\alpha_0 = 0$, \newline $\alpha_1 = \alpha_2 = \alpha_3 = 0.4$, \newline $\alpha_4 = \alpha_5 = \alpha_6 = -0.4$\\
        & Non-Linear & $logit[P(M_B = 1|X=x)] = \alpha_0 + \alpha_1X_1 + \alpha_2X_2^2 + \alpha_3X_3^2 + \alpha_4\exp(X_4) + \alpha_5\sin(X_5) + \alpha_6X_6$ & \\
        \hline
        IPS & Linear & $logit[P(D(1) > D(0)|X=x)] = \theta_0 + \sum_{i=1}^{6}\theta_iX_i$ & $\gamma_0 = 0$, \newline $\gamma_1 = \gamma_2 = \gamma_3 = 0.1$, \newline $\gamma_4 = \gamma_5 = \gamma_6 = -0.1$\\
        & Non-Linear & $logit[P(Z=1|X=x)] = \gamma_0 + \gamma_1X_1 + \gamma_2\exp(X_2) + \gamma_3X_3^2 + \gamma_4X_4 + \gamma_5\cos(X_5) + \gamma_6X_6$ & \\
        \hline
        Compliance Score & Linear &  $logit[P(D(1) > D(0)|X=x)] = \theta_0 + \sum_{i=1}^{6}\theta_iX_i$ & $\theta_0 = 0$, \newline $\theta_1 = \theta_2 = \theta_3 = 0.4$, \newline $\theta_4 = \theta_5 = \theta_6 = -0.8$\\
        & Non-Linear & $logit[P(D(1) > D(0)|X=x)] = \theta_0 + \theta_1X_1^2 + \theta_2\exp(X_2) + \theta_3X_3 + \theta_4X_4^3 + \theta_5X_5 + \theta_6X_6 $& \\
        \hline
        Outcome & Linear & $Y = \beta_1D + \sum_{i=2}^{7}\beta_iDX_i + f(X) + \epsilon$, \newline where $f(X) = \sum_{i=1}^{6}\beta_iX_i$&  $\beta_1 = 1$, \newline $\beta_2 = \beta_3 = \beta_4 = 0.35$, \newline $\beta_5 = \beta_6 = \beta_7 = -0.35$, \newline
        $\beta_8 = \beta_9 = \beta_{10} = 1$, \newline
        $\beta_{11} = \beta_{12} = \beta_{13} = 0.5$
        \newline $\epsilon \sim N(0,1)$ \\
        & Non-Linear & where $f(X) = \beta_8X_1^2 + \beta_9X_2^2 + \beta_{10}\exp(X_3) + \beta_{11}X_4 + \beta_{12}\exp(X_5) + \beta_{13}\cos(X_6)$& \\
        \hline
    \end{tabularx}
\end{table}

\subsection{Results for the Survey Weighted LATE}

Table \ref{tab:simulation_results} details the results over 1000 with the ground truth results in the footnotes. The unweighted LATE estimation and inference is equivalent to that detailed in Kennedy (2023). The simulation results show our SWLATE estimator can successfully replicate the target study LATE across a variety of scenarios as well as estimate in closed form the variance from estimating the weights with our cross-fitting procedure. In the linear setting, regardless of the instrument strength both no SL and SL have minimal absolute bias for the LATE of Study A as well as the desired coverage. In the non-linear setting, the SL approach is able to adapt to non-linearity in the DGM, with almost no bias and the desired coverage. On the other hand, not utilizing SL produces both substantial bias, almost as much as not weighting, and under-coverage. For example, in the moderate IV non linear setting, the absolute relative bias is 12.05\%, almost as much as the unweighted SL. Nevertheless, the coverage of the weighted no SL estimator is better than the unweighted SL at 82.6\% and 45.7\%, respectively, which is largely a reflection of accounting for the estimation of the weights in our model-based estimator. These latter findings are mirrored the both the weak and strong instrument setting.

\begin{table}[h!]
\centering
\caption{Simulation Results for Linear and Non-Linear DGM}
\label{tab:simulation_results}

{\large \textbf{Weak IV}}
\vspace{0.2cm}

\begin{tabular}{lllcccccc}
\toprule
\multicolumn{2}{l}{\textbf{}} & \textbf{} & \multicolumn{1}{p{1.5cm}}{\centering \textbf{Point} \\ \textbf{Estimate}} & \multicolumn{1}{p{1.5cm}}{\centering \textbf{Bias} \\ \textbf{(\%)}} & \multicolumn{1}{p{2cm}}{\centering \textbf{Monte} \\ \textbf{Carlo} \\ \textbf{SE}} & \multicolumn{1}{p{1.5cm}}{\centering \textbf{Model-Based} \\ \textbf{SE}} & \multicolumn{1}{p{1.5cm}}{\centering \textbf{Coverage} \\ \textbf{(95\% CI)}} \\ 
\midrule
\multirow{4}{*}{Linear} & \multirow{2}{*}{Weighted} & No SL & 1.577 & -2.14 & 0.251 & 0.233 & 0.93 \\
& & SL & 1.605 & -0.40 & 0.255 & 0.240 & 0.941 \\
\cmidrule{2-8}
& \multirow{2}{*}{Unweighted} & No SL & 1.790 & 11.02 & 0.172 & 0.167 & 0.811 \\
& & SL & 1.787 & 10.86 & 0.173 & 0.167 & 0.822 \\
\midrule
\multirow{4}{*}{Non-Linear} & \multirow{2}{*}{Weighted} & No SL & 1.762 & 19.80 & 0.300 & 0.287 & 0.843 \\
& & SL & 1.469 & -0.17 & 0.191 & 0.181 & 0.936 \\
\cmidrule{2-8}
& \multirow{2}{*}{Unweighted} & No SL & 1.880 & 27.80 & 0.216 & 0.214 & 0.501 \\
& & SL & 1.636 & 11.23 & 0.141 & 0.137 & 0.773 \\
\bottomrule
\end{tabular}
\caption*{Ground truth: $\beta_{LATE,L}^A = 1.612$, $\beta_{LATE,L}^B = 1.786$; $\beta_{LATE,NL}^A = 1.471$, $\beta_{LATE,NL}^B = 1.636$}

\vspace{0.2cm}

{\large \textbf{Moderate IV}}
\vspace{0.2cm}

\begin{tabular}{lllcccccc}
\toprule
\multicolumn{2}{l}{\textbf{}} & \textbf{} & \multicolumn{1}{p{1.5cm}}{\centering \textbf{Point} \\ \textbf{Estimate}} & \multicolumn{1}{p{1.5cm}}{\centering \textbf{Bias} \\ \textbf{(\%)}} & \multicolumn{1}{p{2cm}}{\centering \textbf{Monte} \\ \textbf{Carlo} \\ \textbf{SE}} & \multicolumn{1}{p{1.5cm}}{\centering \textbf{Model-Based} \\ \textbf{SE}} & \multicolumn{1}{p{1.5cm}}{\centering \textbf{Coverage} \\ \textbf{(95\% CI)}} \\ 
\midrule
\multirow{4}{*}{Linear} & \multirow{2}{*}{Weighted} & No SL & 1.379 & -1.78 & 0.128 & 0.124 & 0.946 \\
& & SL & 1.400 & -0.29 & 0.130 & 0.126 & 0.953 \\
\cmidrule{2-8}
& \multirow{2}{*}{Unweighted} & No SL & 1.603 & 14.17 & 0.100 & 0.095 & 0.431 \\
& & SL & 1.601 & 14.03 & 0.100 & 0.094 & 0.435 \\
\midrule
\multirow{4}{*}{Non-Linear} & \multirow{2}{*}{Weighted} & No SL & 1.431 & 12.05 & 0.166 & 0.162 & 0.826 \\
& & SL & 1.273 & 0.08 & 0.105 & 0.101 & 0.943 \\
\cmidrule{2-8}
& \multirow{2}{*}{Unweighted} & No SL & 1.597 & 25.55 & 0.134 & 0.130 & 0.299 \\
& & SL & 1.447 & 13.76 & 0.087 & 0.083 & 0.457 \\
\bottomrule
\end{tabular}
\caption*{Ground truth: $\beta_{LATE,L}^A = 1.404$, $\beta_{LATE,L}^B = 1.604$; $\beta_{LATE,NL}^A = 1.272$, $\beta_{LATE,NL}^B = 1.447$}

\vspace{0.2cm}

{\large \textbf{Strong IV}}
\vspace{0.2cm}

\begin{tabular}{lllcccccc}
\toprule
\multicolumn{2}{l}{\textbf{}} & \textbf{} & \multicolumn{1}{p{1.5cm}}{\centering \textbf{Point} \\ \textbf{Estimate}} & \multicolumn{1}{p{1.5cm}}{\centering \textbf{Bias} \\ \textbf{(\%)}} & \multicolumn{1}{p{2cm}}{\centering \textbf{Monte} \\ \textbf{Carlo} \\ \textbf{SE}} & \multicolumn{1}{p{1.5cm}}{\centering \textbf{Model-Based} \\ \textbf{SE}} & \multicolumn{1}{p{1.5cm}}{\centering \textbf{Coverage} \\ \textbf{(95\% CI)}} \\ 
\midrule
\multirow{4}{*}{Linear} & \multirow{2}{*}{Weighted} & No SL & 1.219 & -0.51 & 0.085 & 0.084 & 0.951 \\
& & SL & 1.234 & 0.72 & 0.087 & 0.085 & 0.943 \\
\cmidrule{2-8}
& \multirow{2}{*}{Unweighted} & No SL & 1.461 & 19.25 & 0.074 & 0.069 & 0.082 \\
& & SL & 1.460 & 19.21 & 0.074 & 0.069 & 0.084 \\
\midrule
\multirow{4}{*}{Non-Linear} & \multirow{2}{*}{Weighted} & No SL & 1.252 & 10.05 & 0.121 & 0.119 & 0.825 \\
& & SL & 1.136 & -0.15 & 0.075 & 0.074 & 0.949 \\
\cmidrule{2-8}
& \multirow{2}{*}{Unweighted} & No SL & 1.445 & 26.95 & 0.105 & 0.101 & 0.14 \\
& & SL & 1.326 & 16.52 & 0.069 & 0.065 & 0.168 \\
\bottomrule
\end{tabular}
\caption*{Ground truth: $\beta_{LATE,L}^A = 1.225$, $\beta_{LATE,L}^B = 1.461$; $\beta_{LATE,NL}^A = 1.138$, $\beta_{LATE,NL}^B = 1.327$}
\end{table}

Overall, the simulation results show our SWLATE estimator can successfully replicate the target study LATE across a variety of scenarios. Additionally, our efficient, model-based SE is aligns with the Monte Carlo results with and without SL. Clearly, the impact of having to estimate the weights is reflected in the difference between the weighted and unweighted SEs. In the linear DGM, SL incurs only a marginal increase in variance with a similar point estimate while in the non-linear DGM, SL is necessary to prevent bias and undercoverage. 

\subsection{Results for the Weighted ATE Bounds}

For measuring the performance of our weighted ATE bounds, we again use 1000 simulations with the same DGM detailed above with a target ATE of 1. There are three main metrics measured across 1000 simulations: (i) the average of the lower bound estimate across simulations, (ii) the average of of the upper bound estimate across simulations, and (iii) the number of simulations whose intervals that cover the true ATE, which we will colloquially refer to as "coverage." Because the length of the ATE estimate interval is inversely proportional to IV strength, we present the results for strength ranging from $0.1$ to $0.9$. This also ensures that our coverage results are informative since wider intervals are more likely to contain the true ATE.

Figure \ref{fig:ate_results_plot} shows how the average bounds across simulations vary with the IV strength, while Table \ref{tab:ate_results_table} details those results along with the coverage. When the IV is of sufficient strength, the weighted bounds capture the true ATE for the target study at a far larger rate than the unweighted bounds. With a strong IV, the unweighted bounds almost always miss the true ATE while the weighted bounds increasingly center upon the target study ATE. At lower strengths, coverage was more a function of interval width than the result of a weighting. Similar to the LATE results, in the linear setting, SL had little to no impact on the results while in the non-linear setting, SL was crucial to ensuring that the bounds contained the target ATE at increasing strengths.

\begin{figure}[h!]
    \centering
    \includegraphics[scale = 0.4]{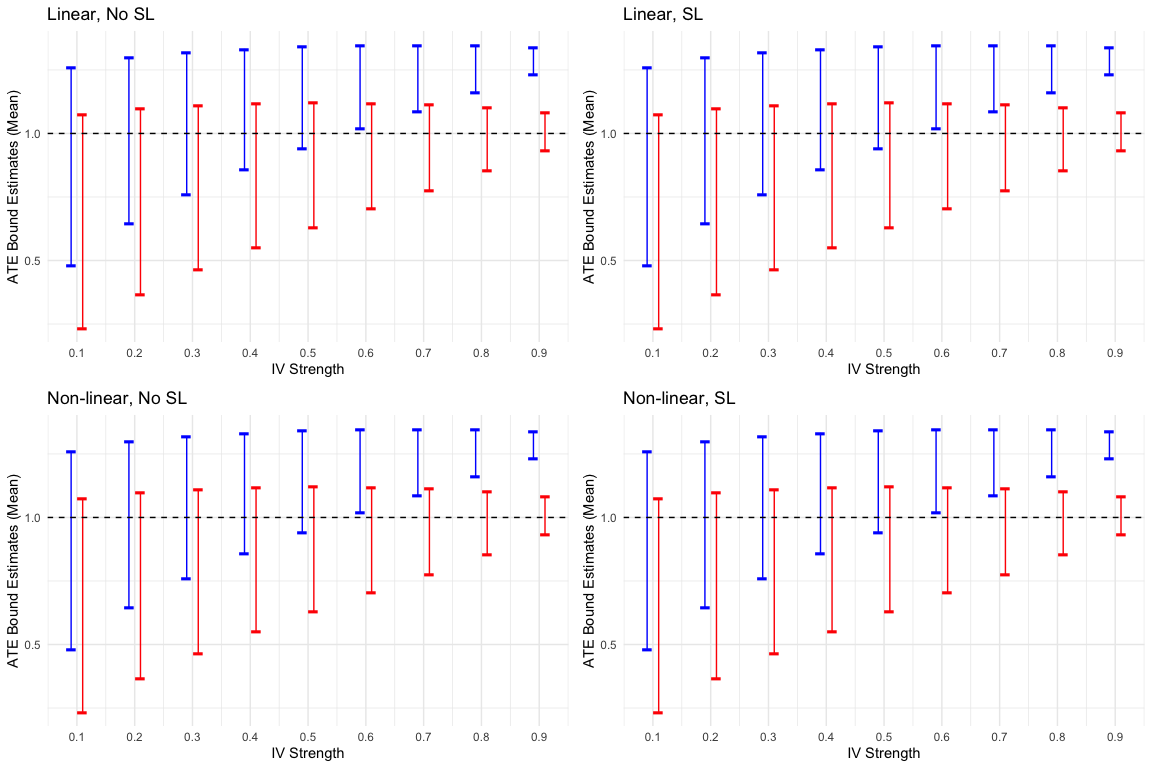}
    \caption{Simulation results from ATE bounds simulation. Blue lines are the unweighted bound estimators while red lines are the weighted bound estimators. The dotted black line indicates the target ATE ground truth of $1$.}
    \label{fig:ate_results_plot}
\end{figure}

\begin{table}[h!]
\centering
\caption{Simulation Results for ATE Bounds by Instrument Strength}
\sisetup{round-mode=places, round-precision=3}
\begin{tabular}{
    >{\bfseries}l 
    S[table-format=1.3] 
    S[table-format=1.3]
    S[table-format=1.3]
    S[table-format=1.3]
    S[table-format=1.3]
    S[table-format=1.3]
    S[table-format=1.3]
    S[table-format=1.3]
    S[table-format=1.3]
}
\toprule
\multicolumn{10}{c}{Linear, No SL} \\
\midrule
{Strength} & \textbf{0.100} & \textbf{0.200} & \textbf{0.300} & \textbf{0.400} & \textbf{0.500} & \textbf{0.600} & \textbf{0.700} & \textbf{0.800} & \textbf{0.900} \\
\midrule
Unweighted LB & 0.483 & 0.646 & 0.762 & 0.857 & 0.941 & 1.018 & 1.088 & 1.160 & 1.234 \\
Weighted LB   & 0.235 & 0.365 & 0.465 & 0.551 & 0.630 & 0.706 & 0.777 & 0.853 & 0.934 \\
\addlinespace
Unweighted UB & 1.261 & 1.299 & 1.320 & 1.333 & 1.342 & 1.348 & 1.348 & 1.347 & 1.339 \\
Weighted UB   & 1.077 & 1.100 & 1.111 & 1.118 & 1.121 & 1.120 & 1.114 & 1.104 & 1.082 \\
\addlinespace
Unweighted Coverage  & 0.999 & 1.000 & 0.996 & 0.965 & 0.779 & 0.385 & 0.108 & 0.009 & 0.000 \\
Weighted Coverage    & 0.867 & 0.931 & 0.935 & 0.940 & 0.949 & 0.947 & 0.941 & 0.915 & 0.723 \\
\midrule
\multicolumn{10}{c}{Linear, SL} \\
\midrule
Unweighted LB    & 0.481 & 0.643 & 0.760 & 0.854 & 0.937 & 1.015 & 1.085 & 1.157 & 1.232 \\
Weighted LB     & 0.244 & 0.373 & 0.471 & 0.557 & 0.634 & 0.710 & 0.780 & 0.856 & 0.937 \\
\addlinespace
Unweighted UB   & 1.258 & 1.296 & 1.316 & 1.330 & 1.339 & 1.346 & 1.345 & 1.345 & 1.339 \\
Weighted UB    & 1.081 & 1.104 & 1.115 & 1.123 & 1.126 & 1.126 & 1.119 & 1.110 & 1.089 \\
\addlinespace
Unweighted Coverage  & 0.999 & 1.000 & 0.996 & 0.965 & 0.794 & 0.406 & 0.132 & 0.009 & 0.000 \\
Weighted Coverage   & 0.867 & 0.928 & 0.937 & 0.939 & 0.955 & 0.948 & 0.937 & 0.904 & 0.721 \\
\midrule
\multicolumn{10}{c}{Non-linear, No SL} \\
\midrule
Unweighted LB    & 0.625 & 0.804 & 0.920 & 1.011 & 1.086 & 1.151 & 1.209 & 1.263 & 1.315 \\
Weighted LB     & 0.407 & 0.562 & 0.667 & 0.753 & 0.827 & 0.893 & 0.954 & 1.014 & 1.074 \\
\addlinespace
Unweighted UB    & 1.340 & 1.379 & 1.390 & 1.396 & 1.396 & 1.395 & 1.392 & 1.386 & 1.380 \\
Weighted UB    & 1.184 & 1.208 & 1.209 & 1.207 & 1.200 & 1.194 & 1.184 & 1.173 & 1.161 \\
\addlinespace
Unweighted Coverage  & 0.999 & 0.959 & 0.769 & 0.445 & 0.215 & 0.071 & 0.021 & 0.007 & 0.005 \\
Weighted Coverage    & 0.952 & 0.965 & 0.963 & 0.960 & 0.904 & 0.772 & 0.632 & 0.404 & 0.212 \\
\midrule
\multicolumn{10}{c}{Non-linear, SL} \\
\midrule
Unweighted LB     & 0.531 & 0.703 & 0.815 & 0.904 & 0.976 & 1.040 & 1.095 & 1.147 & 1.198 \\
Weighted LB     & 0.308 & 0.454 & 0.556 & 0.640 & 0.711 & 0.776 & 0.834 & 0.893 & 0.950 \\
\addlinespace
Unweighted UB   & 1.244 & 1.279 & 1.286 & 1.289 & 1.287 & 1.286 & 1.280 & 1.272 & 1.263 \\
Weighted UB    & 1.078 & 1.102 & 1.101 & 1.098 & 1.091 & 1.083 & 1.073 & 1.059 & 1.043 \\
\addlinespace
Unweighted Coverage  & 0.994 & 0.995 & 0.980 & 0.860 & 0.613 & 0.309 & 0.105 & 0.016 & 0.001 \\
Weighted Coverage  & 0.796 & 0.853 & 0.869 & 0.872 & 0.841 & 0.829 & 0.784 & 0.690 & 0.460 \\
\bottomrule
\end{tabular}
\label{tab:ate_results_table}
\end{table}

\section{Applied Example}

To demonstrate our methodology to replicate IV estimates across studies, we study the effect of triglycerides on cognitive decline by employing MR on data from the Alzheimer's Disease Neuroimaging Initiative (ADNI). ADNI is a multi-center, natural history longitudinal study that tracks cognitive performance among cognitively unimpaired, mild cognitively impaired, or cognitively impaired volunteers aged 55-90 who are in otherwise good health.\citep{petersen_alzheimers_2010} Here we consider two-year change from baseline in the clinical dementia rating sum of boxes (CDR-SB)\citep{obryant_staging_2008} for those with mild cognitive impairment (MCI). Higher CDR-SB scores indicate more severe cognitive decline. The baseline CDR-SB is the earliest visit with recorded MCI. For the final measurement, there must be a CDR-SB measurement 1.75 to 2.5 years after their first measure. Our models both adjust for baseline CDR-SB and time from baseline. Baseline covariates include medical history, age, body mass index (BMI), dementia family history, sex, APOE4 allele count, and years of education.

Using available genetic data, we constructed a binary IV by recording the presence of SNP \textit{rs1260326}, which is associated with elevated triglyceride levels.\citep{willer_discovery_2013} Our exposure, blood triglycerides levels, was binarized into "low" and "high" using the median level in our MCI cohort of 1.05 mmol/L. More details about how the genetic sequencing and lipid measurements were collected and about ADNI in general can be found at \href{https://adni.loni.usc.edu/about/}{https://adni.loni.usc.edu/about/}. Our final cohort had a total sample size of $n = 655$.

To demonstrate our methodology, we synthetically construct two cohorts by sampling the original data with replacement using sampling weights from the model 
\begin{align*}
    logit[P(M_B = 1|X=x)] &= \alpha_0 + \alpha_1 \text{I}_{\text{Cardiovascular Disease}} + \alpha_2 \text{I}_{\text{Neurological Disease}}\\ 
    &+ \alpha_3\text{I}_{\text{Male}} + \alpha_4\text{BMI} + \alpha_5\text{Baseline CDR-SB} + \alpha_6\text{Baseline Age}
\end{align*}
with $\alpha_0 = 0$ and $\alpha_1 = ... = \alpha_6 = 0.1$. Our model indicates those with higher values for the above covariates are more likely to be sampled in $B$. Henceforth, the original data is referred to as the target study and the sampled data the current study. Table \ref{tab:stratified_by_study} details the distributions of covariates across the cohorts. Binary IV strength is measured by the difference of the proportion of high triglyceride level subjects who do and do not have the gene, resulting in 0.13 for the current study and 0.11 for the target study, indicating an instrument of moderate strength.

\begin{table}[h!]
    \centering
    \caption{Baseline distributions of Covariates Stratified by Study}
    \label{tab:stratified_by_study}
    \begin{tabular}{lccc}
    \toprule
        \textbf{Variable} & \textbf{Target Study} & \textbf{Current Study} & \textbf{Standardized} \\
                          & \textbf{(n = 655)} & \textbf{(n = 655)} & \textbf{Mean Difference} \\
        \midrule
        Age (mean (SD)) & 72.96 (7.26) & 71.30 (7.29) & 0.228 \\
        Male (\%) & 394 (60.2) & 400 (61.1) & 0.019 \\
        Years Education (mean (SD)) & 15.98 (2.72) & 15.94 (2.61) & 0.014 \\
        BMI (mean (SD)) & 26.97 (4.77) & 28.75 (5.38) & 0.349 \\
        Family History (\%) & 371 (56.6) & 394 (60.2) & 0.071 \\
        Baseline CDR (mean (SD)) & 1.46 (0.88) & 1.58 (0.97) & 0.139 \\
        Time from Baseline (mean (SD)) & 2.03 (0.08) & 2.03 (0.08) & 0.042 \\
        APOE4 copies (\%) & & & 0.094 \\
        \quad 0 & 332 (50.7) & 314 (47.9) & \\
        \quad 1 & 256 (39.1) & 255 (38.9) & \\
        \quad 2 & 67 (10.2) & 86 (13.1) & \\
        Neurological Condition (\%) & 226 (34.5) & 244 (37.3) & 0.057 \\
        Cardiovascular Condition (\%) & 447 (68.2) & 472 (72.1) & 0.083 \\
        Metabolic Condition (\%) & 275 (0.42) & 288 (0.44) & 0.043 \\
    \bottomrule
    \end{tabular}
\end{table}

The results of our analysis are presented in \ref{tab:ADNI_results}. SL and no SL refer to the same estimators outlined in the previous section with five-fold cross-fitting for estimation. For a baseline ATE estimate, we ran OLS for both studies demonstrate a significant increase in CDR for those with higher triglycerides. To weight the current study OLS estimates, we computed survey weights via logistic regression and used 1000 bootstrap repetitions to calculate the variance. The weighted estimate does move towards the target study estimate but not notably so, indicating an opportunity to better model the weights.

\begin{table}[h!]
\centering
\caption{Estimated Effect of Triglycerides on Cognitive Decline with and without Weighting}
\label{tab:ADNI_results}
\begin{tabularx}{\textwidth}{XlYY} 
\toprule
\textbf{Estimation Method} & \textbf{Target Study} & \multicolumn{2}{c}{\textbf{Current Study}} \\
\cmidrule(lr){3-4}
 & & \textbf{Weighted} & \textbf{Unweighted} \\
\midrule
OLS & 0.34 (0.04, 0.65) & 0.59 (0.29, 0.93) & 0.63 (0.35, 0.93)\\
LATE & & & \\
\hspace{0.2in}No SL & 2.13 (-0.93, 5.56) & 3.53 (-2.46, 9.55) & 5.36 (0.95, 9.77) \\
\hspace{0.2in}SL & 1.44 (-1.7, 4.60) & 1.73 (-4.20, 7.78) & 2.73 (-3.50, 8.98)\\
ATE Bounds & & & \\
\hspace{0.2in}No SL & (-0.95, 0.16) & (-1.00, 0.08) & (-0.81, 0.25) \\
\hspace{0.2in}SL & (-1.11, 0.08) &(-0.93, 0.18) & (-0.77, 0.31)\\
\bottomrule
\end{tabularx}
\end{table}

The target study LATE, with and without SL, both have a positive point estimate but are not statistically significant. The unweighted LATE without SL in the current study shows a significant 5.36-point increase, but after weighting, the estimate attenuates toward the target LATE and becomes insignificant with our weighting method, reducing the relative difference from the target and current estimate from over 150\% to 66\%. SL for the current study LATE produced a lower point estimate and wider confidence intervals than the no SL equivalent (possibly a consequence of using flexible learners at lower sample-sizes) and, thus, was not statistically significant. After weighting the SL current study LATE, we see that the relative difference is reduced from 90\% to 20\%, showing notable attenuation towards the target SL point estimate.

Our ATE bounds in the target study for both no SL and SL suggest a null point estimate with the bounds crossing 0. The bounds for the unweighted current study estimates reach higher, which intuitively follows from the fact that the OLS and LATE estimates were higher. When we apply the weights, the bounds align well with the target study. For no SL, the percentage overlap between the ATE intervals increases from 87\% to 93\% while for SL, the percent overlap increases from 71\% to 84\%.

\section{Discussion}

Crucial to the trustworthiness of causal estimates from observational studies, including those found from MR, is whether they may be replicated across different populations. In this work, we have developed an approach to non-parametrically replicate the LATE across cohorts using possibly unknown survey weights. By focusing on the LATE as opposed to the ATE, our method has the inherent benefit of relaxing the untestable assumptions involved in replicability, including positivity and unmeasured effect modifiers. Explicitly, through the use of ML, our method can protect against functional misspecification of the survey weights. Our empirical results show that for all nuisance functions, using ensemble learning incurs little to no downside even when the DGM is linear, suggesting this approach is an acceptable default when the underlying DGM is unknown. Should we wish to target the ATE, our method performs well on providing bounds on the target ATE despite not making the accompanying causal assumptions, particularly in the case of a stronger IV. These bounds could be further narrowed by incorporating sampling weights with developing work on "covariate-assisted" bounds.\citep{levis_covariate-assisted_2023}

Despite double robustness on most nuisance functions, the survey weights do not share this property. Undoubtedly, the sensitivity to misspecification of the weights depends on the target estimand. Thus, future work could not only extend our work to estimands such as the local average treatment effect on the treated (LATT)\citep{frolich_identification_2013} but furthermore assess scenarios where replicability may be more successful in one estimand compared to other. These results could be used to derive an "optimally weighted" estimator for replicability. In the IV literature, key conditions like monotonicity have been relaxed through similar approaches, trading off "localness" of the estimand for validity.\citep{huntington-klein_instruments_2020} Finally, future work might consider extensions to binary outcomes, continuous IVs, and time-varying treatments.

\section{Appendix}

\subsection{Proof of Theorem \ref{theorem_asymptotic}}
\begin{proof}
    The overall structure of the proof is as follows. We largely follow the strategy outlined in Kennedy (2023) where we must show that $\hat{\beta}^A_{LATE} - \beta^A_{LATE} = S^* + T_1 + T_2$ takes a asymptotically linear form of $(\mathbb{P}_n - P)(\phi(B;w,\zeta_z)) + o_P(n^{-1/2})$.\citep{kennedy_semiparametric_2023} In other words, we have $S^* = (\mathbb{P}_n - P)\{\phi(B;w,\zeta_z)$ and show $T_1 + T_2 = o_P(n^{-1/2})$. We will establish this result for the numerator, whose influence function will be denoted by $\phi$, and skip proof for the denominator as it is similarly a weighted difference of conditional expectations. Once we have done this, we can apply Lemma S.1 from Takatsu et al. (2023), which states the ratio of two asymptotically linear estimators is also asymptotically linear.\citep{takatsu_doubly_2023}

    For simplicity of the proof, we assume that our estimation uses sample-splitting with $K = 2$, but not cross-fitting, as described in Proposition \ref{prop:cross-fit}. Suppose we had an i.i.d. sample from distribution $P_B$: $(B_1,B_2,...,B_{N_B})$ and that $n = \ceil{\frac{N_B}{2}}$, then we fit the nuisance functions with $B^N = (B_n+1,...,B_{N_B})$ and computed the predicted values over $(B_1,...B_n)$ where empirical measure $\mathbb{P}_n$ is over this independent partition. The fact that we are using sample-splitting means that our weights estimator only needs to be consistent for $\frac{w(X)}{E_{P_B}[w(X)]}$. Nevertheless, because study A is involved in estimating the weights, we will take the expectations over $P$, the joint distribution of $P_A$ and $P_B$.
 
    \begin{lemma}\label{lemma:eta_w}
        $|\hat{\eta} - \eta| = o_P(1)$ implies $|\hat{w} - w| = o_P(1)$
    \end{lemma}
    \begin{proof}
        By positivity and the bounding conditions outlined, $\eta,\hat{\eta} > C_{\eta}$ so by a simple arithmetic arrangement we have $|\hat{w} - w| = \bigg|\frac{\eta - \hat{\eta}}{\eta\hat{\eta}}\bigg| \le \frac{1}{C_{\eta}^2}|\hat{\eta} - \eta|$.
    \end{proof}
    
    \begin{lemma}
        $T_1 = o_P(n^{-1/2})$
    \end{lemma}
    \begin{proof}
        Proving $T_1 = o_P(n^{-1/2})$ is an application of results of Kennedy et al (2020) Lemma 2, which states $T_1 = O_P\bigg(\frac{\norm{\hat{\phi} - \phi}}{n^{-1/2}}\bigg)$ due to sample-splitting.\citep{kennedy_sharp_2020} Therefore, we meet the condition for $T_1$ as long as $\norm{\hat{\phi} - \phi} = o_P(1)$. We can investigate this condition in more detail for $T(P_B) = \frac{E_{P_B}[w(X)E[Y|Z=1,X]]}{E_{P_B}[w(X)]}$, which naturally extends to our estimand:
    \begin{align}
        \hat{\phi} - \phi &= \frac{\hat{w}}{\mathbb{P}_n\hat{w}}\bigg\{\hat{\mu}_1 + \frac{Z(Y-\hat{\mu}_1)}{\hat{e}}\bigg\} - \frac{w}{E_{P_B}[w]}\bigg\{\mu_1 + \frac{Z(Y-\mu_1)}{e}\bigg\}\\
        &= \frac{\hat{w}\hat{\mu_1}}{\mathbb{P}_n\hat{w}}\big(1-\frac{Z}{\hat{e}}\big) - \frac{w\mu_1}{E_{P_B}[w]}\big(1-\frac{Z}{e}\big) + \frac{ZY}{\hat{e}e}\big(\frac{\hat{w}\hat{e}}{\mathbb{P}_n\hat{w}} - \frac{we}{E_{P_B}[w]}\big)\\
        &\le \big(\frac{\hat{w}\hat{\mu}_1}{\mathbb{P}_n\hat{w}} - \frac{w\mu_1}{E_{P_B}[w]}\big)\big(1+\frac{1}{C_e}\big) + \frac{C_Y}{{C_e}^2}\big(\frac{\hat{w}\hat{e}}{\mathbb{P}_n\hat{w}} - \frac{we}{E_{P_B}[w]}\big)
    \end{align}
    Therefore, taking the $L_2$ norm, we will have
    \begin{align}
        \norm{\hat{\phi} - \phi} \le \norm{\frac{\hat{w}\hat{\mu}_1}{\mathbb{P}_n\hat{w}} - \frac{w\mu_1}{E_{P_B}[w]}}\big(1+\frac{1}{C_e}\big) + \frac{C_Y}{{C_e}^2}\norm{\frac{\hat{w}\hat{e}}{\mathbb{P}_n\hat{w}} - \frac{we}{E_{P_B}[w]}} \label{eq:condition_t1_l2}
    \end{align}
    Thus, it is sufficient that $\norm{\frac{\hat{w}\hat{\mu}_1}{\mathbb{P}_n\hat{w}} - \frac{w\mu_1}{E_{P_B}[w]}} = o_P(1)$ and $\norm{\frac{\hat{w}\hat{e}}{\mathbb{P}_n\hat{w}} - \frac{we}{E_{P_B}[w]}} = o_P(1)$ for the whole term to be $o_P(1)$. This is achieved as long as we have consistency of the first term of each $L_2$ norm to the second term of each $L_2$ norm and smoothness of the various nuisance functions.\citep{chernozhukov_doubledebiased_2018,kennedy_semiparametric_2023} Using sample splitting, we can completely avoid Donsker conditions regarding the complexity of the nuisance functions. For $\hat{\mu}_1$, this is implied by assumption. In the following steps, we will show consistency for $\hat{w}$ and $\mathbb{P}_n\hat{w}$.
    \end{proof}
    
    To show  $|\mathbb{P}_n\hat{w} - E_{P_B}[w]| = o_P(1)$, we will proceed by Markov's inequality, noting that we have estimated $w$ via sample splitting and, therefore, conditioning on $B^N$ and $A$, the data from study A, will yield $\hat{w}$ fixed.
    \begin{align}
        P(|\mathbb{P}_n\hat{w} - E_{P_B}[w]| &\ge \epsilon) = E[P(|\mathbb{P}_n\hat{w} - E_{P_B}[w]| \ge \epsilon|\big|B^N,A)]\\
        &\le \epsilon^{-1}E[|\mathbb{P}_n\hat{w} - E_{P_B}[w]|\big|B^N,A]\\
        &= \epsilon^{-1}E[|\mathbb{P}_n\hat{w} - E_{P_B}[w] + \mathbb{P}_nw - \mathbb{P}_nw|\big|B^N,A]\\
        &\le \epsilon^{-1}E[|\hat{w} - w|] + E[|\mathbb{P}_nw - E_{P_B}[w]|] = o(1)
    \end{align}
    \noindent where the last inequality follows by a combination of triangle inequality and the fact that
    \begin{align*}
        E[|\mathbb{P}_n(\hat{w} - w)|\big|B^N,A] \le  n^{-1}\sum_iE\big[|\hat{w}(x_i) - w(X)|\big|B^N,A\big],
    \end{align*}
    \noindent an i.i.d. sum, because $w$ is now fixed, similar to the proof of Kennedy et al (2020) Lemma 2.\citep{kennedy_sharp_2020} Now, the first term goes to 0 by $|\hat{w} - w| = o_P(1)$ and uniform integrability because $\hat{w}$ and $w$ are bounded. The second term goes to 0 by weak law of large numbers and and uniform integrability due to boundedness.

    Now that we have established all the terms in \ref{eq:condition_t1_l2} are consistent, the whole term is consistent by Slutsky's theorem. Thus, we we will have $L_2$ convergence due to boundedness and, consequentially, $T_1 = O_P\bigg(\frac{\norm{\hat{\phi} - \hat{\phi}}}{n^{-1/2}}\bigg) = o_P(n^{-1/2})$.

    \begin{lemma}
        $T_2 = o_P(n^{-1/2})$
    \end{lemma}
    \begin{proof}
    For $T_2$, we must derive the remainder term $R_2(P,\hat{P}) = E_P[\hat{T}_{\text{1-step}}] - T(P_B)$ for the numerator. We will begin first with $T(P_B) = \frac{E_{P_B}[w(X)E[Y|Z=1,X]]}{E_{P_B}[w(X)]}$:
    \begin{align}
        R_2(P,\hat{P}) &= \bigintssss \frac{\hat{w}}{\mathbb{P}_n\hat{w}}\bigg[\frac{Z}{\hat{e}}(Y-\hat{\mu}_1) + \hat{\mu}_1 \bigg] dP - \frac{E_{P_B}[w\mu_1]}{E_{P_B}[w]} \\
        &= \bigintssss \frac{\hat{w}}{\mathbb{P}_n\hat{w}} \bigg[\frac{\mu_1 e}{\hat{e}} -\frac{\hat{\mu}_1}{\hat{e}} + \hat{\mu}_1\bigg] dP - \frac{E_{P_B}[w\mu_1]}{E_{P_B}[w]}\\
        &= \bigintssss \frac{\hat{w}}{\mathbb{P}_n\hat{w}} \bigg[\frac{-\hat{e}}{e}(\hat{\mu}_1 - \mu_1) + \hat{\mu}_1 + (\mu_1 - \mu_1)\bigg] dP - \frac{E_{P_B}[w\mu_1]}{E_{P_B}[w]} \\
        &= \bigintssss \frac{\hat{w}}{\mathbb{P}_n\hat{w}} \bigg[\frac{1}{\hat{e}}(\hat{e} - e)(\hat{\mu}_1 - \mu_1)\bigg]dP + \bigintssss \mu_1 \bigg[\frac{\hat{w}}{\mathbb{P}_n\hat{w}} - \frac{w}{E_{P_B}[w]}\bigg]dP
    \end{align}
    \noindent where the second equality follows by iterated expectation and the last follows because $P$ is a joint distribution that includes $P_B$. Thus, taking the absolute value by bounding assumptions, triangle equality, and Cauchy-Schwarz, we have
    \begin{align}
        |R_2(P,\hat{P})| &\le \frac{1}{C_{\hat{e}}}\bigintssss |\frac{\hat{w}}{\mathbb{P}_n\hat{w}}||\hat{e} - e||\hat{\mu}_1 - \mu_1| dP + C_{\mu_1}\bigintssss\big|\frac{\hat{w}}{\mathbb{P}_n\hat{w}} - \frac{w}{E_{P_B}[w]} \big|dP\\
        &\le \frac{C_w}{C_{\hat{e}}}\norm{\hat{e} - e}\norm{\hat{\mu}_1 - \mu_1} + C_{\mu_1}\bigintssss\bigg|\frac{\hat{w}}{\mathbb{P}_n\hat{w}} - \frac{w}{E_{P_B}[w]} \bigg|dP\\
        &= o_P(n^{-1/2}) + o_P(1).
    \end{align}
    The first term is $o_P(n^{-1/2})$ by boundedness of $\norm{\frac{\hat{w}}{\mathbb{P}_n\hat{w}}}$ and the product $\norm{\hat{e} - e}\norm{\hat{\mu}_1 - \mu_1} = o_P(n^{-1/2})$, which holds if $\norm{\hat{e} - e} = o_P(n^{-1/4})$ and $\norm{\hat{\mu}_1 - \mu_1} = o_P(n^{-1/4})$. Now we will focus on the second term, where we will show $L_1$ convergence.
    
    Firstly, we have $|\hat{w} - w| = o_P(1)$ by Lemma \ref{lemma:eta_w}. Furthermore, when proving the rate of $T_1$, we showed that $|\mathbb{P}_n\hat{w} - E_{P_B}[w]| = o_P(1)$. Thus, by Slutsky's theorem we have that $|\frac{\hat{w}}{\mathbb{P}_n\hat{w}} - \frac{w}{E_{P_B}[w]}| = o_P(1)$. Given that we have uniform integrability due to boundedness of $\frac{\hat{w}}{\mathbb{P}_n\hat{w}}$ so we have $L_1$ convergence or, in other words, $E\bigg[\big|\frac{\hat{w}}{\mathbb{P}_n\hat{w}} - \frac{w}{E_{P_B}[w]}\big|\bigg] = o_P(1)$.
    \end{proof}
    
    The proof for $T({P_B}) = \frac{E_{P_B}[w(X)\{\mu_1-\mu_0\}]}{E_{P_B}[w(X)]}$ is the same as the one above where we repeat the process for $\mu_1$ and $\mu_0$ because it takes the form of a WATE. The proof for the denominator mirrors that of the numerator. Now we are ready to prove convergence and derive the asymptotic variance of our estimator.

    From the result of Takatsu Lemma S.1, assuming there exists $\epsilon > 0$ such that $|\phi_{num}(\mathbf{B};\hat{w},\hat{\zeta_z})|\wedge|\phi_{denom}(\mathbf{B};\hat{w},\hat{\zeta_z})| > \epsilon$, then we have the following asymptotically linear form
    \begin{align*}
        \frac{\mathbb{P}_n\hat{\phi}_{num}}{\mathbb{P}_n\hat{\phi}_{num}} - \frac{P\phi_{num}}{P\phi_{denom}} &= \mathbb{P}_n\bigg\{\bigg(\frac{E_{P_B}[w\{m_1(X) - m_0(X)\}]}{E[w(X)]}\bigg)^{-1}\bigg(\mathbb{IF}_{num}-\beta^A_{LATE}\mathbb{IF}_{denom}\bigg)\bigg\} + o_P(n^{-1/2})\\
        &= \mathbb{P}_n\bigg\{\bigg(\frac{E_{P_B}[w\{m_1(X) - m_0(X)\}]}{E_{P_B}[w(X)]}\bigg)^{-1}\bigg(\frac{2Z-1}{e(X,Z)}\frac{w(X)}{E[w(X)]}\{Y-\mu_z(X)\} + \frac{w\{\mu_1(X) - \mu_0(X)\}}{E_{P_B}[w(X)]}\\
        &- \beta^A_{LATE}\bigg[\frac{2Z-1}{e(X,Z)}\frac{w(X)}{E_{P_B}[w(X)]}\{D-m_z(X)\} + \frac{w\{m_1(X) - m_0(X)\}}{E_{P_B}[w(X)]}\bigg]\bigg)\bigg\} + o_P(n^{-1/2})\\
        &= \mathbb{P}_n\Bigg(\frac{w(X)}{E_{P_B}[w(X)\{m_1(X) - m_0(X)\}]}\bigg\{\frac{2Z-1}{e(X,Z)}\bigg[Y-\mu_Z(X)-\beta^A_{LATE}\{D-m_Z(X)\}\bigg]\\
        &+ \mu_1(X) - \mu_0(X) - \beta^A_{LATE}\{m_1(X) - m_0(X)\}\Bigg) + o_P(n^{-1/2})
    \end{align*}
    \noindent where $e(X,Z) = e(X)Z + \{1-e(X)\}(1-Z)$. We can observe this is a sample mean of the influence function for the weighted LATE and, thus, after multiplying each side by $n^{-1/2}$, we have that $n^{-1/2}(\hat{\beta}^A_{LATE} - \beta^A_{LATE}) \overset{d}{\to} N(0,E_{P_B}[\Gamma^2])$ where
    \begin{align*}
        \Gamma &= \frac{w(X)}{E_{P_B}[w(X)\{m_1(X) - m_0(X)\}]}\bigg\{\frac{2Z-1}{e(X,Z)}\bigg[Y-\mu_Z(X)-\beta^A_{LATE}\{D-m_Z(X)\}\bigg]\\
        &+ \mu_1(X) - \mu_0(X) - \beta^A_{LATE}\{m_1(X) - m_0(X)\}\bigg\}
    \end{align*}
    \noindent via Slutsky's theorem and the standard central limit theorem, giving us $\sqrt{n}$-convergence.
\end{proof}

\section{Acknowledgements and Data Avaliability}

RSZ is supported by NIH/NIA T32
AG073088-03. DLG is supported by NIH/NIA 5P30AG066519 and 1RF1AG075107.

Data used in the preparation of this article were obtained from the Alzheimer’s Disease Neuroimaging Initiative (ADNI) database (adni.loni.usc.edu). The ADNI was launched in 2003 as a public-private partnership, led by Principal Investigator Michael W. Weiner, MD. The primary goal of ADNI has been to test whether serial magnetic resonance imaging (MRI), positron emission tomography (PET), other biological markers, and clinical and neuropsychological assessment can be combined to measure the progression of mild cognitive impairment (MCI) and early Alzheimer’s disease (AD). The ADNI dataset is avaliable to the public at http://adni.loni.usc.edu via the completion of an online application form and acceptance of data use agreement. All participants in ADNI provided informed consent. For details regarding the ethics statement of the ADNI study population please see: https://adni.loni.usc.edu.

Data collection and sharing for this project was funded by the Alzheimer's Disease
Neuroimaging Initiative (ADNI) (National Institutes of Health Grant U01 AG024904) and
DOD ADNI (Department of Defense award number W81XWH-12-2-0012). ADNI is funded
by the National Institute on Aging, the National Institute of Biomedical Imaging and
Bioengineering, and through generous contributions from the following: AbbVie, Alzheimer’s
Association; Alzheimer’s Drug Discovery Foundation; Araclon Biotech; BioClinica, Inc.;
Biogen; Bristol-Myers Squibb Company; CereSpir, Inc.; Cogstate; Eisai Inc.; Elan
Pharmaceuticals, Inc.; Eli Lilly and Company; EuroImmun; F. Hoffmann-La Roche Ltd and
its affiliated company Genentech, Inc.; Fujirebio; GE Healthcare; IXICO Ltd.; Janssen
Alzheimer Immunotherapy Research \& Development, LLC.; Johnson \& Johnson
Pharmaceutical Research \& Development LLC.; Lumosity; Lundbeck; Merck \& Co., Inc.;
Meso Scale Diagnostics, LLC.; NeuroRx Research; Neurotrack Technologies; Novartis
Pharmaceuticals Corporation; Pfizer Inc.; Piramal Imaging; Servier; Takeda Pharmaceutical
Company; and Transition Therapeutics. The Canadian Institutes of Health Research is
providing funds to support ADNI clinical sites in Canada. Private sector contributions are
facilitated by the Foundation for the National Institutes of Health (www.fnih.org). The grantee
organization is the Northern California Institute for Research and Education, and the study is
coordinated by the Alzheimer’s Therapeutic Research Institute at the University of Southern
California. ADNI data are disseminated by the Laboratory for Neuro Imaging at the
University of Southern California.

\bibliographystyle{unsrtnat}
\bibliography{references}

\end{document}